\documentclass[a4paper]{article}
\usepackage{a4wide}
\usepackage{amsmath}
\usepackage{amssymb}
\usepackage{amsthm}
\usepackage{enumerate}
\usepackage{fancyhdr}
\usepackage{color}
\usepackage{ifpdf}
\usepackage{graphicx}
\usepackage{tikz}
\usetikzlibrary{calc}
\usepackage{rotating}
\usepackage[T2A]{fontenc}
\usepackage{mathtext}
\usepackage{caption}
\newtheorem{thm}{Theorem}
\newtheorem{prop}{Proposition}

\newtheorem{lemma}{Lemma}
\newcommand{\bbr}{{\mathbb R}}
\newcommand{\bbs}{{\mathbb S}}
\newcommand{\+}{|\!|\!|}

\usepackage{authblk}

\begin{document}
\title{Asymptotic behaviour of the Boltzmann equation as a cosmological model}

\author{Ho Lee\footnote{holee@khu.ac.kr}}

\affil{Department of Mathematics and Research Institute for Basic Science, Kyung Hee University, Seoul, 130-701, Republic of Korea}

\maketitle

\begin{abstract}
As a Newtonian cosmological model the Vlasov-Poisson-Boltzmann system is considered, and a slightly modified Boltzmann equation, which describes the stability of an expanding universe, is derived. Asymptotic behaviour of solutions turns out to depend on the expansion of the universe, and in this paper we consider the soft potential case and will obtain asymptotic behaviour.
\end{abstract}

\section{Introduction}
A standard way to describe the dynamics of an ensemble of particles in a cosmological scale is to use general relativity and kinetic theory. The Einstein-Vlasov or Einstein-Boltzmann systems are commonly used to study the dynamics of many different kinds of cosmological models, for instance the FLRW, de Sitter, Bianchi spacetimes, etc. In this paper we are interested in the dynamics of a cosmological model, but Newtonian gravity theory will be considered instead of general relativity. The Vlasov-Poisson-Boltzmann(VPB) system will be the most relevant Newtonian model to study the time evolution of particles, where the particles interact with each other through gravity and collisions, and this system may be considered as a non-relativistic version of the Einstein-Boltzmann system. We refer to \cite{Lee132,LeeNungesser,LeeRendall131,LeeRendall132} for the Einstein-Boltzmann system, and in the present paper we study the VPB system as a cosmological model.

The VPB system is a system of partial differential equations which can describe the evolution of matter distribution in a statistical way such that the matter is treated as a collection of particles. In this system gravity and collisions are taken into account as the interactions between particles, and in particular Poisson's equation describes the gravitational attraction between particles. In a cosmological scale, however, the expansion of the universe must be considered in addition to the gravity. Einstein's equations can describe in an integrated way the expansion of the universe and the gravitational attraction between particles, but Poisson's equation does not. To describe the expansion of the universe in the context of Newtonian cosmology, we follow the approach of \cite{Lee10,Lee131,Rein97,ReinRendall94}. The authors of \cite{ReinRendall94} first considered the Vlasov-Poisson(VP) system as a cosmological model. An exact solution of the VP system, which is a description of an expanding universe, was introduced, and then existence, asymptotic behaviour, and stability of solutions were studied for its perturbations in \cite{Lee10,Rein97,ReinRendall94}. This approach was applied to the VPB system in \cite{Lee131}, where existence was proved in the hard sphere case, but asymptotic behaviour was not treated. It was observed that the effect of collisions by the Boltzmann equation decreases in a cosmological setting, and in particular it seems to decrease more in the hard sphere case. In this paper, instead, we consider the soft potential case and will see that the effect of collisions is recovered to an extent. In the end, we will obtain asymptotic behaviour, and furthermore will see that the softer the potential is, the more rapidly solutions decay.

\subsection{The VPB system in a cosmological setting}
The VPB system is written as follows:
\begin{align*}
&\partial_t F+ v \cdot \nabla_x F -\nabla_x\phi\cdot\nabla_v F=\int_{\bbr^3\times\bbs^2}|u-v|^\gamma B(\omega)\Big(F(u')F(v')-F(u)F(v)\Big)d\omega du,\\
&\Delta \phi=4\pi\rho,\quad\rho=\int_{\bbr^3}F dv,
\end{align*}
where $F=F(t,x,v)$ is the velocity distribution function and describes the density of particles at $x\in\bbr^3$ with velocity $v\in\bbr^3$ at time $t\geq 0$. The quantities $\phi=\phi(t,x)$ and $\rho=\rho(t,x)$ are the gravitational potential and the mass density, respectively, and the right side of the first equation is called the collision operator with post-collision velocities $v'$ and $u'$ for given pre-collision velocities $v$ and $u$, where
\[
v'=v+((u-v)\cdot\omega)\omega,\quad
u'=u+((v-u)\cdot\omega)\omega,\quad\omega\in\bbs^2.
\]
The scattering kernel in the collision operator has an exponent $\gamma$, which ranges $-3<\gamma\leq 1$, and the Grad angular cutoff assumption, i.e., $0<B(\omega)\leq C|\cos\theta|$ with the scattering angle $\theta$, will be assumed throughout the paper. We refer to \cite{Glassey} for more details on the Boltzmann equation.

We consider the following solution to the VPB system, which was introduced in \cite{Lee131}:
\begin{align}\label{universe}
(\mu,\rho_0,\phi_0):=\Big(\pi^{-\frac32}\exp\Big(-a^2(t)|v-\dot{a}(t)a^{-1}(t)x|^2\Big),a^{-3}(t),(2\pi/3)a^{-3}(t)|x|^2\Big),
\end{align}
where the scale factor $a(t)$ satisfies the following ordinary differential equation:
\begin{align}\label{scale_factor}
\ddot{a}(t)+(4\pi/3)a^{-2}(t)=0,
\end{align}
and the above quantities are now understood as a universe, which is spatially homogeneous and isotropic. Depending on the asymptotic behaviour of the scale factor, the above quantities may describe an expanding or a shrinking universes, and this will be discussed later. We now take a perturbation to \eqref{universe} as $F=\mu+\sqrt{\mu}f$, $\rho=\rho_0+\sigma$, and $\phi=\phi_0+\psi$ to get a system of equations for perturbations $(f,\sigma,\psi)$, and then take the transformation $\bar{x}=a^{-1}x$ and $\bar{v}=av-\dot{a}x$ as in \cite{Lee131}, to obtain a system of equations for $(\bar{f},\bar{\sigma},\bar{\psi})$ such that $f(t,x,v)=\bar{f}(t,\bar{x},\bar{v})$, $\sigma(t,x)=\bar{\sigma}(t,\bar{x})$, and $\psi(t,x)=\bar{\psi}(t,\bar{x})$. After dropping the bars we get the following equations:
\begin{align}
&\partial_tf+a^{-2}v\cdot\nabla_xf-\nabla_x\psi\cdot\nabla_vf+fv\cdot\nabla_x\psi+2\sqrt{\mu}v\cdot\nabla_x\psi\nonumber\allowdisplaybreaks\\
&\quad=a^{-3-\gamma}\iint |u-v|^\gamma B(\omega)e^{-\frac{1}{2}|u|^2}\Big(e^{-\frac{1}{2}|v'|^2}f(u')+e^{-\frac{1}{2}|u'|^2}f(v')\Big)d\omega du\nonumber\\
&\qquad-a^{-3-\gamma}\iint |u-v|^\gamma B(\omega)e^{-\frac{1}{2}|u|^2}\Big(e^{-\frac{1}{2}|v|^2}f(u)+e^{-\frac12|u|^2}f(v)\Big)d\omega du\nonumber\\
&\qquad+a^{-3-\gamma}\iint |u-v|^\gamma B(\omega)e^{-\frac{1}{2}|u|^2}\Big(f(u')f(v')-f(u)f(v)\Big)d\omega du,\allowdisplaybreaks\label{vpb1_tilde}\\
&\Delta \psi=4\pi a^{-1}\sigma=4\pi a^{-1}\int e^{-\frac{1}{2}|v|^2}f dv,\label{vpb2_tilde}
\end{align}
and the universe model \eqref{universe} is transformed to $\mu(v)=\pi^{-3/2}\exp(-|v|^2)$, $\rho_0(t)=a^{-3}(t)$, and $\phi_0(t,x)=(2\pi/3)a^{-1}(t)|x|^2$. Note that the quantity $\mu(v)$ is now identical to the usual Maxwellian.

\subsection{The Boltzmann equation as a cosmological model}
As a simplest cosmological model we consider the spatially homogeneous case of the equations \eqref{vpb1_tilde}--\eqref{vpb2_tilde}:
\begin{align}\label{boltz}
\partial_tf+a^{-3-\gamma}Lf=a^{-3-\gamma}\Gamma(f,f),
\end{align}
where $L=\nu-K$ with $K=K_2-K_1$ such that
\begin{align*}
&K_2f=\iint |u-v|^\gamma B(\omega)e^{-\frac{1}{2}|u|^2}\Big(e^{-\frac{1}{2}|v'|^2}f(u')+e^{-\frac{1}{2}|u'|^2}f(v')\Big)d\omega du,\\
&K_1f=e^{-\frac{1}{2}|v|^2}\iint |u-v|^\gamma B(\omega)e^{-\frac{1}{2}|u|^2}f(u) d\omega du,\allowdisplaybreaks\\
&\nu=\iint |u-v|^\gamma B(\omega)e^{-|u|^2}d\omega du=\mbox{``collision frequency''},\\
&\Gamma(f,f)=\iint |u-v|^\gamma B(\omega)e^{-\frac{1}{2}|u|^2}\Big(f(u')f(v')-f(u)f(v)\Big)d\omega du,
\end{align*}
and the equation \eqref{vpb2_tilde} vanishes.

In this paper we will study the equation \eqref{boltz} with \eqref{scale_factor}. Global existence was considered in \cite{Lee131} in the hard sphere case, i.e., $\gamma=1$, but their asymptotic behaviour was not obtained. In the present paper we consider the soft potential case, i.e., $-3<\gamma<0$, and will obtain asymptotic behaviour of solutions of \eqref{boltz}. Note that the homogeneity assumption is taken on the transformed equations \eqref{vpb1_tilde}--\eqref{vpb2_tilde}. Hence, the solutions of \eqref{boltz} will be inhomogeneous in the original coordinates, and the result will show that the solutions converge to the universe model \eqref{universe} in an inhomogeneous way.

\section{Preliminaries}
Let $\langle\cdot,\cdot\rangle$ and $\|\cdot\|$ be the standard inner product and the corresponding $L^2$-norm, respectively, in $\bbr^3_v$, and define $\langle g,h\rangle_\nu=\langle \nu g,h\rangle$ with the corresponding norm $\|\cdot\|_\nu$. For a multi-index $\beta$, we write $\partial_\beta f(t,v)=\partial_v^{\beta}f(t,v)$, and throughout the paper $N$ will denote a fixed positive integer such that $N\geq 4$. For a non-negative integer $k$, we define
\begin{align*}
\+f(t)\+_k^2=\sum_{|\beta|\leq N}\|w^{|\beta|-k}\partial_\beta f(t)\|^2,\quad\+f(t)\+_{\nu,k}^2=\sum_{|\beta|\leq N}\|w^{|\beta|-k}\partial_\beta f(t)\|_\nu^2,
\end{align*}
where $w(v)=(1+|v|)^\gamma$ is a weight function. Note that $w$ is equivalent to the collision frequency. For simplicity we will write $a_\gamma(t):=a^{-3-\gamma}(t)$, and to prove existence we will need the following quantities:
\[
E_k(t)=\frac{1}{2}\+f(t)\+_k^2+\int_0^t a_\gamma(s)\+f(s)\+_{\nu,k}^2ds,\quad\mathcal{E}_m(t)=\sum_{k=0}^m E_k(t).
\]
Since $w\leq 1$, we have $E_k\leq E_m$ for $k\leq m$, and can see that $\mathcal{E}_m$ and $E_m$ are equivalent. The following quantity will be used to study asymptotic behaviour:
\[
y_r(t)=\sum_{0\leq k\leq r}\+f(t)\+_k^2=\sum_{\substack{|\beta|\leq N\\ 0\leq k\leq r}}\|w^{|\beta|-k}\partial_\beta f(t)\|^2.
\]
Let $a$ be a solution of the equation \eqref{scale_factor}. Then, the quantity $E_a=(1/2)\dot{a}^2-(4\pi/3)a^{-1}$ is conserved, which is non-negative for the expanding cases. According to \cite{ReinRendall94}, we have two cases: if $E_a=0$, or equivalently $\dot{a}(0)=(8\pi/3)^{1/2}$, then we solve $a(t)=(\sqrt{6\pi}t+1)^{2/3}$, while if $E_a>0$, or equivalently $\dot{a}(0)>(8\pi/3)^{1/2}$, then we see that $2E_a\leq \dot{a}(t)^2\leq 2E_a+8\pi/3$, and in this sense we say that $a(t)\sim t^{2/3}$ or $a(t)\sim t$. Here, we set $a(0)=1$ for simplicity.

The following lemmas are originally introduced in \cite{Guo031}, but the statements and conditions of the lemmas are slightly improved by the arguments of \cite{StrainGuo06} as follows. We refer to \cite{Guo031,StrainGuo06} for the proofs of the following lemmas.
\begin{lemma}\label{lem1}
For any $\theta\in\bbr$, $|\langle w^{2\theta} Kg_1,g_2\rangle|\leq C\|w^\theta g_1\|_\nu \|w^\theta g_2\|_\nu$.
\end{lemma}

\begin{lemma}\label{lem2}
For any $\theta\in\bbr$ and $\eta>0$, the operator $K$ can be split into $K=K_c+K_s$, where $K_c$ is a compact operator in $L^2_\nu(\bbr^3)$ with respect to $|\cdot|_\nu$ and $|\langle w^{2\theta} K_sg_1,g_2\rangle|\leq \eta\|w^\theta g_1\|_\nu \|w^\theta g_2\|_\nu$.
\end{lemma}

\begin{lemma}\label{lem_positivity}
There exists a $\delta>0$ such that $\langle Lg,g\rangle\geq \delta\|g\|_\nu^2$.
\end{lemma}

\begin{lemma}\label{lem4}
For any $\theta\in\bbr$, $\beta\geq 0$, $\eta>0$, and $k\geq 0$, there exist $C_\eta>0$ and $C_k>0$ satisfying the following estimates:
\begin{align}
&\langle w^{2\theta}\partial_\beta[\nu g],\partial_\beta g\rangle\geq\|w^{\theta}\partial_\beta g\|_\nu^2-\Big(\eta\sum_{|\beta_1|\leq|\beta|}\|w^\theta\partial_{\beta_1}g\|_\nu^2+C_\eta\|w^\theta g\|_\nu^2\Big),\label{lem_est_linear1}\\
&|\langle w^{2\theta}\partial_\beta[Kg_1],\partial_\beta g_2\rangle|\leq\Big(\eta\sum_{|\beta_1|\leq|\beta|}\|w^\theta\partial_{\beta_1}g_1\|_\nu+C_\eta\|w^\theta g_1\|_\nu\Big)\|w^\theta\partial_\beta g_2\|_\nu,\label{lem_est_linear2}\\
&\langle w^{-2k}L g, g\rangle\geq\frac12\|w^{-k} g\|_\nu^2-C_k\| g\|_\nu^2.\label{lem_est_linear3}
\end{align}
\end{lemma}

\begin{lemma}\label{lem_est_gamma}
For any $\theta\in\bbr$, $\beta\geq 0$, and $k\geq 0$, the following estimates hold:
\begin{align}
&|\langle w^{2\theta}\partial_\beta \Gamma_+(g_1,g_2),\partial_\beta g_3\rangle|\leq C\Big(\sum_{|\beta_1|\leq N}\|w^\theta\partial_{\beta_1}g_1\|_\nu\Big)\Big(\sum_{|\beta_2|\leq N}\|w^\theta\partial_{\beta_2}g_2\|_\nu\Big)\|w^\theta\partial_\beta g_3\|_\nu\label{lem_est_gamma1}\\
&\quad\quad + C\Big(\sum_{|\beta_1|\leq N}\|w^{|\beta_1|-k}\partial_{\beta_1}g_1\|\Big)\Big(\sum_{|\beta_1|+|\beta_2|\leq |\beta|}\|w^{\theta-|\beta_1|+k}\partial_{\beta_2}g_2\|_\nu\Big)\|w^\theta\partial_\beta g_3\|_\nu,\nonumber\\
&|\langle w^{2\theta}\partial_\beta \Gamma_-(g_1,g_2),\partial_\beta g_3\rangle|\leq C\Big(\sum_{|\beta_1|\leq N}\|w^\theta\partial_{\beta_1}g_1\|_\nu\Big)\Big(\sum_{|\beta_2|\leq N}\|w^\theta\partial_{\beta_2}g_2\|_\nu\Big)\|w^\theta\partial_\beta g_3\|_\nu.\label{lem_est_gamma2}
\end{align}
\end{lemma}
\begin{proof}
The proof of this lemma is basically the same with the one in \cite{Guo031}, where the estimate of $\Gamma$ was considered for $\theta\geq 0$, $k=0$, and $N\geq 8$. To make it clear that the lemma still holds for any $\theta\in\bbr$ with $k\geq 0$ and $N\geq 4$, we briefly present the proof of this lemma. Here, we only consider the derivatives with respect to $v$. Note that the gain and the loss terms are written as follows:
\begin{align*}
\partial_\beta \Gamma_\pm(g_1,g_2)&=\sum_{\beta_0+\beta_1+\beta_2=\beta}\Gamma^0_\pm(\partial_{\beta_1}g_1,\partial_{\beta_2}g_2),
\end{align*}
where
\begin{align*}
\Gamma^0_+(\partial_{\beta_1}g_1,\partial_{\beta_2}g_2)&=\int |u-v|^\gamma B(\omega)\partial_{\beta_0}\Big[e^{-\frac12 |u|^2}\Big](\partial_{\beta_1}g_1)(u')(\partial_{\beta_2}g_2)(v')d\omega du,\\
\Gamma^0_-(\partial_{\beta_1}g_1,\partial_{\beta_2}g_2)&=\int |u-v|^\gamma B(\omega)\partial_{\beta_0}\Big[e^{-\frac12 |u|^2}\Big](\partial_{\beta_1}g_1)(u)(\partial_{\beta_2}g_2)(v)d\omega du.
\end{align*}
Below, we estimate the gain and the loss terms separately. For $k\geq 2$, we know that
\[
H^k(\bbr^3)\subset C(\bbr^3)\cap L^\infty(\bbr^3),
\]
and this will be frequently used in the estimates.\\

\noindent{\bf Loss term.} We first prove the estimate for the loss term. Since $\beta_0+\beta_1+\beta_2=\beta$ for $|\beta|\leq N$, we have either (a) $|\beta_1|\leq N/2$ or (b) $|\beta_2|\leq N/2$. For (a), we note that
\begin{align*}
&|\Gamma^0_-(\partial_{\beta_1}g_1,\partial_{\beta_2}g_2)|\\
&\quad\leq C\int |u-v|^\gamma e^{-\frac14|u|^2}|\partial_{\beta_1}g_1(u)||\partial_{\beta_2}g_2(v)|du\\
&\quad\leq C\Big(\sup_u e^{-\frac18|u|^2}|\partial_{\beta_1}g_1(u)|\Big)\bigg(\int |u-v|^\gamma e^{-\frac18|u|^2}du\bigg)|\partial_{\beta_2}g_2(v)|\allowdisplaybreaks\\
&\quad\leq C\Big(\sum_{|\beta_1|\leq N}\|w^\theta\partial_{\beta_1}g_1\|_\nu\Big)(1+|v|)^\gamma|\partial_{\beta_2}g_2(v)|.
\end{align*}
Therefore, the loss term in the case of (a) is estimated as follows:
\begin{align*}
&|\langle w^{2\theta}\Gamma^0_-(\partial_{\beta_1}g_1,\partial_{\beta_2}g_2),\partial_\beta g_3\rangle|\\
&\quad\leq C\int w^{2\theta}(v)\Big(\sum_{|\beta_1|\leq N}\|w^\theta\partial_{\beta_1}g_1\|_\nu\Big)(1+|v|)^\gamma|\partial_{\beta_2}g_2(v)||\partial_\beta g_3(v)| dv\allowdisplaybreaks\\
&\quad\leq C\Big(\sum_{|\beta_1|\leq N}\|w^\theta\partial_{\beta_1}g_1\|_\nu\Big)\\
&\quad\quad\times\bigg(\int w^{2\theta}(v)(1+|v|)^\gamma|\partial_{\beta_2}g_2(v)|^2dv\bigg)^{\frac12}\bigg(\int w^{2\theta}(v)(1+|v|)^\gamma|\partial_{\beta}g_3(v)|^2dv\bigg)^{\frac12}\allowdisplaybreaks\\
&\quad\leq C\Big(\sum_{|\beta_1|\leq N}\|w^\theta\partial_{\beta_1}g_1\|_\nu\Big)\|w^\theta\partial_{\beta_2}g_2\|_\nu\|w^\theta\partial_{\beta}g_3\|_\nu\\
&\quad\leq C\Big(\sum_{|\beta_1|\leq N}\|w^\theta\partial_{\beta_1}g_1\|_\nu\Big)\Big(\sum_{|\beta_2|\leq N}\|w^\theta\partial_{\beta_2}g_2\|_\nu\Big)\|w^\theta\partial_\beta g_3\|_\nu,
\end{align*}
and summing over all the possible multi-indices $\beta_0$, $\beta_1$, and $\beta_2$, we have the estimate \eqref{lem_est_gamma2}.
For (b), we decompose the integration domain $\bbr^6$ for $u$ and $v$ into three cases: (b.1) $|u-v|\leq |v|/2$, (b.2) $|u-v|\geq |v|/2$ with $|v|\geq 1$, and (b.3) $|u-v|\geq |v|/2$ with $|v|\leq 1$. For (b.1), we use
\[
|u|\geq \frac{1}{2}|v|\quad\mbox{and}\quad e^{-\frac{1}{8}|u|^2}\leq e^{-\frac{1}{32}|v|^2}\quad\mbox{for (b.1)}.
\]
Then, we have
\begin{align}
&|\langle w^{2\theta}\Gamma^0_-(\partial_{\beta_1}g_1,\partial_{\beta_2}g_2),\partial_\beta g_3\rangle|\nonumber\\
&\quad\leq C\int w^{2\theta}(v)|u-v|^\gamma e^{-\frac{1}{4}|u|^2}|\partial_{\beta_1}g_1(u)||\partial_{\beta_2}g_2(v)||\partial_{\beta}g_3(v)|du dv\nonumber\\
&\quad\leq C\Big(\sup_v e^{-\frac{1}{64}|v|^2}|\partial_{\beta_2}g_2(v)|\Big)\int w^{2\theta}(v)|u-v|^\gamma e^{-\frac{1}{8}|u|^2}e^{-\frac{1}{64}|v|^2}|\partial_{\beta_1}g_1(u)||\partial_{\beta}g_3(v)|dudv\nonumber\allowdisplaybreaks\\
&\quad\leq C \Big(\sum_{|\beta_2|\leq N}\|w^\theta\partial_{\beta_2}g_2\|_\nu\Big)\nonumber\\
&\quad\quad\times\bigg(\int w^{2\theta}(v)|u-v|^\gamma e^{-\frac{1}{8}|u|^2}e^{-\frac{1}{32}|v|^2}|\partial_{\beta_1}g_1(u)|^2 dudv\bigg)^{\frac{1}{2}}\nonumber\\
&\quad\quad\times\bigg(\int w^{2\theta}(v)|u-v|^\gamma e^{-\frac{1}{8}|u|^2}|\partial_{\beta} g_3(v)|^2 dudv\bigg)^{\frac{1}{2}}\nonumber\allowdisplaybreaks\\
&\quad\leq C\Big(\sum_{|\beta_2|\leq N}\|w^\theta\partial_{\beta_2}g_2\|_\nu\Big)\nonumber\\
&\quad\quad\times\bigg(\int (1+|u|)^\gamma e^{-\frac{1}{8}|u|^2}|\partial_{\beta_1}g_1(u)|^2 du\bigg)^{\frac{1}{2}}
\bigg(\int w^{2\theta}(v)(1+|v|)^\gamma |\partial_{\beta} g_3(v)|^2 dv\bigg)^{\frac{1}{2}}\nonumber\\
&\quad\leq C\Big(\sum_{|\beta_2|\leq N}\|w^\theta\partial_{\beta_2}g_2\|_\nu\Big)\|w^{\theta} \partial_{\beta_1}g_1\|_\nu \|w^{\theta} \partial_\beta g_3\|_\nu,\label{lem_est_gamma_b.1}
\end{align}
where we used the fact that $N\geq 4$ and $|\beta_2|\leq N/2$ in the third inequality.
For (b.2), we note that since $\gamma$ is negative,
\[
|u-v|^\gamma\leq |v|^\gamma \leq C(1+|v|)^\gamma\quad\mbox{for (b.2).}
\]
Then, we have
\begin{align}
&|\langle w^{2\theta}\Gamma^0_-(\partial_{\beta_1}g_1,\partial_{\beta_2}g_2),\partial_\beta g_3\rangle|\nonumber\\
&\quad\leq C\int w^{2\theta}(v)|u-v|^\gamma e^{-\frac{1}{4}|u|^2}|\partial_{\beta_1}g_1(u)||\partial_{\beta_2}g_2(v)||\partial_{\beta}g_3(v)|dudv\nonumber\allowdisplaybreaks\\
&\quad\leq C\int w^{2\theta}(v)(1+|v|)^\gamma e^{-\frac{1}{4}|u|^2}|\partial_{\beta_1}g_1(u)||\partial_{\beta_2}g_2(v)||\partial_{\beta}g_3(v)|dudv\nonumber\\
&\quad\leq C\bigg(\int e^{-\frac{1}{4}|u|^2}|\partial_{\beta_1}g_1(u)|du\bigg)\int w^{2\theta}(v)(1+|v|)^\gamma |\partial_{\beta_2}g_2(v)||\partial_{\beta}g_3(v)|dv\nonumber\allowdisplaybreaks\\
&\quad\leq C\bigg(\int e^{-\frac{1}{4}|u|^2}du\bigg)^{\frac12}\bigg(\int e^{-\frac{1}{4}|u|^2}|\partial_{\beta_1}g_1(u)|^2du\bigg)^{\frac12}\nonumber\\
&\quad\quad\times\bigg(\int w^{2\theta}(v)(1+|v|)^\gamma |\partial_{\beta_2}g_2(v)|^2dv\bigg)^{\frac12}
\bigg(\int w^{2\theta}(v)(1+|v|)^\gamma |\partial_{\beta}g_3(v)|^2dv\bigg)^{\frac12}\nonumber\\
&\quad\leq C\|w^{\theta}\partial_{\beta_1}g_1\|_\nu \|w^{\theta}\partial_{\beta_2}g_2\|_\nu \|w^{\theta}\partial_{\beta}g_3\|_\nu.\label{lem_est_gamma_b.2}
\end{align}
For (b.3), we use
\[
|u-v|^\gamma\leq C|v|^\gamma\quad\mbox{and}\quad C\leq (1+|v|)^\gamma\leq 1\quad\mbox{for (b.3)},
\]
and
\begin{align*}
&\int |u-v|^{\frac{\gamma}{2}} e^{-\frac{1}{4}|u|^2}|\partial_{\beta_1}g_1(u)|du\\
&\quad\leq \bigg(\int e^{-\frac14|u|^2}|u-v|^\gamma du\bigg)^{\frac12}\bigg(\int e^{-\frac14|u|^2}|\partial_{\beta_1}g_1(u)|^2 du\bigg)^{\frac12}\\
&\quad\leq C \|w^{\theta}\partial_{\beta_1}g_1\|_\nu.
\end{align*}
The estimate for (b.3) case is given as follows:
\begin{align}
&|\langle w^{2\theta}\Gamma^0_-(\partial_{\beta_1}g_1,\partial_{\beta_2}g_2),\partial_\beta g_3\rangle|\nonumber\\
&\quad\leq C\int w^{2\theta}(v)|u-v|^\gamma e^{-\frac{1}{4}|u|^2}|\partial_{\beta_1}g_1(u)||\partial_{\beta_2}g_2(v)||\partial_{\beta}g_3(v)|dudv\nonumber\\
&\quad\leq C\|w^{\theta}\partial_{\beta_1}g_1\|_\nu\int w^{2\theta}(v)|v|^\frac{\gamma}{2} |\partial_{\beta_2}g_2(v)||\partial_{\beta}g_3(v)|dv\nonumber\allowdisplaybreaks\\
&\quad\leq C\|w^{\theta}\partial_{\beta_1}g_1\|_\nu\bigg(\int w^{2\theta}(v)|v|^{\gamma} |\partial_{\beta_2}g_2(v)|^2dv\bigg)^{\frac12}\bigg(\int w^{2\theta}(v)|\partial_{\beta}g_3(v)|^2dv\bigg)^{\frac12}\nonumber\\
&\quad\leq C\|w^{\theta}\partial_{\beta_1}g_1\|_\nu\Big(\sup_vw^{\theta}(v)(1+|v|)^{\frac{\gamma}{2}} |\partial_{\beta_2}g_2(v)|\Big)\bigg(\int |v|^{\gamma} dv\bigg)^{\frac12}\|w^{\theta}\partial_\beta g_3\|_\nu\nonumber\\
&\quad\leq C\|w^{\theta}\partial_{\beta_1}g_1\|_\nu \|w^{\theta}\partial_{\beta_2}g_2\|_\nu\|w^{\theta}\partial_\beta g_3\|_\nu,\label{lem_est_gamma_b.3}
\end{align}
where we used the fact that $(1+|v|)^{\gamma}$ is bounded from below. We now combine \eqref{lem_est_gamma_b.1}--\eqref{lem_est_gamma_b.3} and sum over all the possible multi-indices to obtain the estimate \eqref{lem_est_gamma2}. This completes the proof of the case (b) and the proof of the estimate for the loss term as well.\\

\noindent{\bf Gain term.} For the estimate of the gain term we decompose the integration domain $\bbr^6$ for $u$ and $v$ into three parts: (c.1) $|u|\geq |v|/2$, (c.2) $|u|\leq |v|/2$ with $|v|\geq 1$, and (c.3) $|u|\leq |v|/2$ with $|v|\leq 1$. For (c.1), we use
\[
e^{-\frac18|u|^2}\leq e^{-\frac{1}{32}|v|^2}\quad\mbox{for (c.1).}
\]
Then, we have the following estimate:
\begin{align}
&|\langle w^{2\theta}\Gamma^0_+(\partial_{\beta_1}g_1,\partial_{\beta_2}g_2),\partial_\beta g_3\rangle|\nonumber\\
&\quad\leq C\int w^{2\theta}(v)|u-v|^\gamma B(\omega) e^{-\frac{1}{4}|u|^2}|\partial_{\beta_1}g_1(u')||\partial_{\beta_2}g_2(v')||\partial_{\beta}g_3(v)|d\omega dudv\nonumber\\
&\quad\leq C\int w^{2\theta}(v)|u-v|^\gamma e^{-\frac{1}{8}|u|^2} e^{-\frac{1}{32}|v|^2}|\partial_{\beta_1}g_1(u')||\partial_{\beta_2}g_2(v')||\partial_{\beta}g_3(v)|d\omega dudv\nonumber\allowdisplaybreaks\\
&\quad\leq C\bigg(\int w^{2\theta}(v)|u-v|^\gamma e^{-\frac{1}{8}|u|^2} e^{-\frac{1}{16}|v|^2}|\partial_{\beta_1}g_1(u')|^2|\partial_{\beta_2}g_2(v')|^2d\omega dudv\bigg)^{\frac12}\nonumber\\
&\quad\quad\times\bigg(\int w^{2\theta}(v)|u-v|^\gamma e^{-\frac{1}{8}|u|^2} |\partial_{\beta}g_3(v)|^2d\omega dudv\bigg)^{\frac12}\nonumber\allowdisplaybreaks\\
&\quad\leq C\bigg(\int |u-v|^\gamma e^{-\frac{1}{32}|u|^2} e^{-\frac{1}{32}|v|^2}|\partial_{\beta_1}g_1(u')|^2|\partial_{\beta_2}g_2(v')|^2d\omega dudv\bigg)^{\frac12}\label{lem_est_gamma_c.1.1}\\
&\quad\quad\times\bigg(\int w^{2\theta}(v)(1+|v|)^\gamma |\partial_{\beta}g_3(v)|^2 dv\bigg)^{\frac12},\label{lem_est_gamma_c.1.2}
\end{align}
where the last inequality holds for any $\theta\in\bbr$. Note that the quantity \eqref{lem_est_gamma_c.1.2} is equal to $\|w^\theta\partial_\beta g_3\|_\nu$. For the quantity \eqref{lem_est_gamma_c.1.1}, we use the change of variables $dudv=du'dv'$ and may assume that $|\beta_1|\leq N/2$ without loss of generality. Then, \eqref{lem_est_gamma_c.1.1} is estimated as follows:
\begin{align}
&\bigg(\int |u-v|^\gamma e^{-\frac{1}{32}|u|^2} e^{-\frac{1}{32}|v|^2}|\partial_{\beta_1}g_1(u')|^2|\partial_{\beta_2}g_2(v')|^2d\omega dudv\bigg)^{\frac12}\nonumber\\
&\quad\leq C\bigg(\int |u-v|^\gamma e^{-\frac{1}{32}|u|^2} e^{-\frac{1}{32}|v|^2}|\partial_{\beta_1}g_1(u)|^2|\partial_{\beta_2}g_2(v)|^2dudv\bigg)^{\frac12}\nonumber\\
&\quad\leq C\Big(\sup_u e^{-\frac{1}{128}|u|^2}|\partial_{\beta_1}g_1(u)|\Big)\bigg(\int |u-v|^\gamma e^{-\frac{1}{64}|u|^2} e^{-\frac{1}{32}|v|^2}|\partial_{\beta_2}g_2(v)|^2dudv\bigg)^{\frac12}\nonumber\allowdisplaybreaks\\
&\quad\leq C\Big(\sum_{|\beta_1|\leq N}\|w^\theta \partial_{\beta_1}g_1\|_\nu\Big)\bigg(\int (1+|v|)^\gamma e^{-\frac{1}{32}|v|^2}|\partial_{\beta_2}g_2(v)|^2dv\bigg)^{\frac12}\nonumber\\
&\quad\leq C\Big(\sum_{|\beta_1|\leq N}\|w^\theta \partial_{\beta_1}g_1\|_\nu\Big)\|w^\theta\partial_{\beta_2}g_2\|_\nu,\nonumber
\end{align}
and we conclude for the (c.1) case,
\begin{align}
&|\langle w^{2\theta}\Gamma^0_+(\partial_{\beta_1}g_1,\partial_{\beta_2}g_2),\partial_\beta g_3\rangle|\nonumber\\
&\quad\leq C\Big(\sum_{|\beta_1|\leq N}\|w^\theta \partial_{\beta_1}g_1\|_\nu\Big)\|w^\theta\partial_{\beta_2}g_2\|_\nu\|w^\theta\partial_\beta g_3\|_\nu.\label{lem_est_gamma_c.1}
\end{align}
For (c.2), we use the following basic inequalities:
\begin{gather*}
|u-v|\geq \frac12 |v|,\quad |v|^\gamma\leq C(1+|v|)^\gamma,\quad |u'|+|v'|\leq C(|u|+|v|)\leq C|v|,\\
(1+|v|)^\gamma\leq C(1+|u'|)^\gamma,\quad(1+|v|)^\gamma\leq C(1+|v'|)^\gamma,
\end{gather*}
to get the following estimate:
\begin{align}
&|\langle w^{2\theta}\Gamma^0_+(\partial_{\beta_1}g_1,\partial_{\beta_2}g_2),\partial_\beta g_3\rangle|\nonumber\\
&\quad\leq C\int w^{2\theta}(v)|u-v|^\gamma B(\omega) e^{-\frac{1}{4}|u|^2}|\partial_{\beta_1}g_1(u')||\partial_{\beta_2}g_2(v')||\partial_{\beta}g_3(v)|d\omega dudv\nonumber\allowdisplaybreaks\\
&\quad\leq C\bigg(\int w^{2\theta}(v)|u-v|^\gamma |\partial_{\beta_1}g_1(u')|^2|\partial_{\beta_2}g_2(v')|^2d\omega du dv\bigg)^{\frac12}\nonumber\\
&\qquad\times\bigg(\int w^{2\theta}(v)|u-v|^\gamma e^{-\frac12|u|^2}|\partial_{\beta}g_3(v)|^2d\omega dudv\bigg)^{\frac12}\nonumber\allowdisplaybreaks\\
&\quad\leq C\bigg(\int w^{2\theta}(v)(1+|v|)^\gamma |\partial_{\beta_1}g_1(u')|^2|\partial_{\beta_2}g_2(v')|^2d\omega du dv\bigg)^{\frac12}\nonumber\\
&\qquad\times\bigg(\int w^{2\theta}(v)(1+|v|)^\gamma |\partial_{\beta}g_3(v)|^2dv\bigg)^{\frac12}\nonumber\allowdisplaybreaks\\
&\quad\leq C\bigg(\int w^{2|\beta_1|-2k}(u')w^{2\theta-2|\beta_1|+2k}(v')(1+|v'|)^\gamma |\partial_{\beta_1}g_1(u')|^2|\partial_{\beta_2}g_2(v')|^2d\omega du dv\bigg)^{\frac12}\nonumber\\
&\qquad\times\|w^\theta \partial_\beta g_3\|_\nu\nonumber\allowdisplaybreaks\\
&\quad\leq C\|w^{|\beta_1|-k}\partial_{\beta_1}g_1\|\|w^{\theta-|\beta_1|+k}\partial_{\beta_2}g_2\|_\nu\|w^\theta \partial_\beta g_3\|_\nu.\label{lem_est_gamma_c.2}
\end{align}
For (c.3), we use the following inequalities:
\begin{gather*}
|u|\leq\frac12,\quad  C\leq (1+|v|)^\gamma\leq C(1+|u|)^\gamma,\quad |u-v|\geq \frac12 |v|,\\
|u'|+|v'|\leq C(|u|+|v|)\leq C|v|,\quad (1+|v|)^\gamma\leq C(1+|u'|)^\gamma,\quad(1+|v|)^\gamma\leq C(1+|v'|)^\gamma,
\end{gather*}
to get the following estimate:
\begin{align}
&|\langle w^{2\theta}\Gamma^0_+(\partial_{\beta_1}g_1,\partial_{\beta_2}g_2),\partial_\beta g_3\rangle|\nonumber\\
&\quad\leq C\int w^{2\theta}(v)|u-v|^\gamma B(\omega) e^{-\frac{1}{4}|u|^2}|\partial_{\beta_1}g_1(u')||\partial_{\beta_2}g_2(v')||\partial_{\beta}g_3(v)|d\omega dudv\nonumber\allowdisplaybreaks\\
&\quad\leq C\int |u-v|^\gamma |\partial_{\beta_1}g_1(u')||\partial_{\beta_2}g_2(v')||\partial_{\beta}g_3(v)|d\omega dudv\nonumber\allowdisplaybreaks\\
&\quad\leq C\bigg(\int |v|^\gamma |\partial_{\beta_1}g_1(u')|^2|\partial_{\beta_2}g_2(v')|^2d\omega du dv\bigg)^{\frac12}\bigg(\int |u-v|^\gamma |\partial_{\beta}g_3(v)|^2d\omega dudv\bigg)^{\frac12}.\nonumber
\end{align}
Since we have $|v|^\gamma\leq C|u'|^\gamma$ and $|v|^\gamma\leq C|v'|^\gamma$, we may assume that $|\beta_1|\leq N/2$ without loss of generality. Then, the first integral on the right hand side is estimated as follows:
\begin{align}
&\bigg(\int |v|^\gamma |\partial_{\beta_1}g_1(u')|^2|\partial_{\beta_2}g_2(v')|^2d\omega du dv\bigg)^{\frac12}\nonumber\\
&\quad\leq C\bigg(\int |u'|^\gamma |\partial_{\beta_1}g_1(u')|^2|\partial_{\beta_2}g_2(v')|^2d\omega du dv\bigg)^{\frac12}\nonumber\allowdisplaybreaks\\
&\quad\leq C\bigg(\int |u|^\gamma |\partial_{\beta_1}g_1(u)|^2|\partial_{\beta_2}g_2(v)|^2 du dv\bigg)^{\frac12}\nonumber\allowdisplaybreaks\\
&\quad\leq C\Big(\sup_u w^\theta(u)(1+|u|)^{\frac{\gamma}{2}}|\partial_{\beta_1}g_1(u)|\Big)\bigg(\int |u|^\gamma |\partial_{\beta_2}g_2(v)|^2 du dv\bigg)^{\frac12}\nonumber\allowdisplaybreaks\\
&\quad\leq C\Big(\sum_{|\beta_1|\leq N}\|w^\theta\partial_{\beta_1}g_1\|_\nu\Big)\bigg(\int |\partial_{\beta_2}g_2(v)|^2 dv\bigg)^{\frac12}\nonumber\\
&\quad\leq C\Big(\sum_{|\beta_1|\leq N}\|w^\theta\partial_{\beta_1}g_1\|_\nu\Big)\|w^\theta\partial_{\beta_2}g_2\|_\nu,\label{lem_est_gamma_c.3.1}
\end{align}
and the second integral is estimated as
\begin{align}
&\bigg(\int |u-v|^\gamma |\partial_{\beta}g_3(v)|^2d\omega dudv\bigg)^{\frac12}\leq C\bigg(\int |\partial_{\beta}g_3(v)|^2dv\bigg)^{\frac12}\leq C\|w^\theta\partial_{\beta}g_3\|_\nu.\label{lem_est_gamma_c.3.2}
\end{align}
We combine \eqref{lem_est_gamma_c.3.1} and \eqref{lem_est_gamma_c.3.2} to conclude for the (c.3) case
\begin{align}
&|\langle w^{2\theta}\Gamma^0_+(\partial_{\beta_1}g_1,\partial_{\beta_2}g_2),\partial_\beta g_3\rangle|\nonumber\\
&\quad\leq C\Big(\sum_{|\beta_1|\leq N}\|w^\theta\partial_{\beta_1}g_1\|_\nu\Big)\|w^\theta\partial_{\beta_2}g_2\|_\nu\|w^\theta\partial_{\beta}g_3\|_\nu.\label{lem_est_gamma_c.3}
\end{align}
We finally combine \eqref{lem_est_gamma_c.1}, \eqref{lem_est_gamma_c.2}, and \eqref{lem_est_gamma_c.3} and sum over all the possible multi-indices $\beta_0$, $\beta_1$, and $\beta_2$ to obtain the estimate \eqref{lem_est_gamma1}. This completes the proof of the estimate for the gain term.
\end{proof}

\section{Main results}
We apply the well-known methods of \cite{Guo031,StrainGuo06} to prove global existence and to obtain asymptotic behaviour. The results will show that the universe model \eqref{universe} is asymptotically stable.
\subsection{Global-in-time existence}
Local-in-time existence is easily proved by the energy method of \cite{Guo031}: We consider a standard iteration $\{f^n\}$ and apply Lemmas \ref{lem1}, \ref{lem2}, \ref{lem4}, and \ref{lem_est_gamma} appropriately, then obtain a differential inequality for $\|w^{|\beta|-k}\partial_\beta f^n\|$, where the smallness of $\eta$ may be ignored. For each $k\geq 0$, summing over all $|\beta|\leq N$, we obtain an estimate of $d(\+f^{n}\+_k^2/2)/dt+a_\gamma\+f^{n}\+_{\nu,k}^2$, and then show that the iteration $\{E^n_k\}$ is uniformly bounded on a (short) time interval, where $E^n_k$ is defined by $E_k$ with $f^n$ in place of $f$. As in \cite{Guo031}, the local-in-time existence is proved, and uniqueness, continuity of $E_k$, and the non-negativity of $F$ are also similarly proved. We obtain the following result.
\begin{lemma}\label{lem_local}
Let $k\geq 0$ be an integer, $f_0$ be an initial data of the equation \eqref{boltz} satisfying $F_0=\mu+\sqrt{\mu}f_0\geq 0$, and $a$ be a solution of the equation \eqref{scale_factor} satisfying $E_a\geq 0$. Then, there exist a constant $M>0$ such that if $E_k(0)<M/2$, then the equation \eqref{boltz} admits a unique solution $f$ on a time interval $[0,T]$, corresponding to the initial data $f_0$, such that $F=\mu+\sqrt{\mu}f\geq 0$ and $\sup_{[0,T]}E_k(t)\leq M$, where $E_k$ is continuous on $[0,T]$.
\end{lemma}
The local-in-time solution of Lemma \ref{lem_local} is now extended to a global-in-time solution. Let $f$ be a local solution, and consider the equation \eqref{boltz}. Taking $\partial_\beta$ and multiplying $w^{2|\beta|-2k}\partial_\beta f$ to the equation, we have
\begin{align}
&w^{2|\beta|-2k}\partial_\beta f\partial_t\partial_\beta f+a_\gamma w^{2|\beta|-2k}\partial_\beta \Big[\nu(v)f\Big]\partial_\beta f-a_\gamma w^{2|\beta|-2k}\partial_\beta \Big[ Kf\Big]\partial_\beta f\nonumber\\
&\quad=a_\gamma w^{2|\beta|-2k}\partial_\beta\Gamma(f,f)\partial_\beta f,\label{global}
\end{align}
and then integrate the resulting equation over $\bbr^3_v$. We first apply Lemma \ref{lem_est_gamma} to the right hand side and can see that it is estimated as $a_\gamma |\langle w^{2|\beta|-2k}\partial_\beta\Gamma(f,f),\partial_\beta f\rangle|\leq C a_\gamma E^{1/2}_k\+ f\+_{\nu,k}^2$, which holds for any $\beta\geq 0$ and $k\geq 0$. To estimate the other quantities, we separate the cases as $k=0$ and $k>0$.\\

\noindent{\bf Case 1.} ($k=0$). Apply Lemma \ref{lem_positivity} for $\beta=0$, and use \eqref{lem_est_linear1} and \eqref{lem_est_linear2} for $0<|\beta|\leq N$. Summing over all the possible $0\leq|\beta|\leq N$ and using the smallness of $\eta>0$, we obtain
\begin{align}
\frac12\frac{d}{dt}\sum_{|\beta|\leq N}\|w^{|\beta|}\partial_\beta f\|^2+\delta a_\gamma\sum_{|\beta|\leq N}\|w^{|\beta|}\partial_\beta f\|_\nu^2\leq Ca_\gamma E_k^{\frac12}\+f\+_{\nu,k}^2,\label{case1}
\end{align}
where the constant $\delta>0$ is the one given in Lemma \ref{lem_positivity}.\\

\noindent{\bf Case 2.} ($k>0$). We first consider the case $\beta=0$. Applying \eqref{lem_est_linear3} to \eqref{global} and summing over all $0<k\leq m$, we obtain
\begin{align}\label{case21}
\frac12\frac{d}{dt}\sum_{0<k\leq m}\|w^{-k}f\|^2+\frac{a_\gamma}{2}\sum_{0<k\leq m}\|w^{-k}f\|_\nu^2\leq C_ma_\gamma \|f\|_\nu^2+Ca_\gamma E_m^{\frac12}\+f\+_{\nu,m}^2,
\end{align}
where we used the fact that $\+\cdot\+_k\leq\+\cdot\+_m$ and $\+\cdot\+_{\nu,k}\leq\+\cdot\+_{\nu,m}$ for $k\leq m$. For the case $\beta\neq 0$, we apply \eqref{lem_est_linear1} and \eqref{lem_est_linear2}, and sum over all the possible $0<|\beta|\leq N$ and $0<k\leq m$, to get
\begin{align*}
&\frac12\frac{d}{dt}\sum_{\substack{0<|\beta|\leq N\\0<k\leq m}}\|w^{|\beta|-k}\partial_\beta f\|^2+a_\gamma\sum_{\substack{0<|\beta|\leq N\\0<k\leq m}}\|w^{|\beta|-k}\partial_\beta f\|_\nu^2\\
&\quad \leq a_\gamma\eta\sum_{\substack{0<|\beta|\leq N\\0<k\leq m}}\|w^{|\beta|-k}\partial_\beta f\|_\nu^2+C_\eta a_\gamma\sum_{0<k\leq m}\|w^{-k}f\|_\nu^2+Ca_\gamma E_m^{\frac12}\+f\+_{\nu,m}^2,
\end{align*}
which shows that for a small $\eta>0$,
\begin{align}
&\frac12\frac{d}{dt}\sum_{\substack{0<|\beta|\leq N\\0<k\leq m}}\|w^{|\beta|-k}\partial_\beta f\|^2+\frac{a_\gamma}{2}\sum_{\substack{0<|\beta|\leq N\\0<k\leq m}}\|w^{|\beta|-k}\partial_\beta f\|_\nu^2\nonumber\\
&\quad \leq C_\eta a_\gamma\sum_{0<k\leq m}\|w^{-k}f\|_\nu^2+Ca_\gamma E_m^{\frac12}\+f\+_{\nu,m}^2.\label{case22}
\end{align}
Multiply the equation \eqref{case21} by a suitable constant and add it to the equation \eqref{case22}, then the first quantity on the right side of \eqref{case22} is controlled, and we obtain
\begin{align}
&\frac12\frac{d}{dt}\sum_{\substack{|\beta|\leq N\\0<k\leq m}}\|w^{|\beta|-k}\partial_\beta f\|^2+\frac{a_\gamma}{2}\sum_{\substack{|\beta|\leq N\\0<k\leq m}}\|w^{|\beta|-k}\partial_\beta f\|_\nu^2 \leq C a_\gamma\|f\|_\nu^2+Ca_\gamma E_m^{\frac12}\+f\+_{\nu,m}^2,\label{case2}
\end{align}
where the constants $C$ may now depend on $\eta$ and $m$.\\

\noindent{\bf Case 3.} ($k\geq 0$). Multiply the equation \eqref{case1} by a suitable constant and add it to the equation \eqref{case2}, then the first quantity on the right side of \eqref{case2}, and we obtain
\begin{align*}
\frac12\frac{d}{dt}\sum_{\substack{|\beta|\leq N\\0\leq k\leq m}}\|w^{|\beta|-k}\partial_\beta f\|^2+\delta a_\gamma\sum_{\substack{|\beta|\leq N\\0\leq k\leq m}}\|w^{|\beta|-k}\partial_\beta f\|_\nu^2 \leq Ca_\gamma E_m^{\frac12}\+f\+_{\nu,m}^2.
\end{align*}
Multiplying $\delta^{-1}$ to the both sides, with the definition of $\mathcal{E}_m$, we get
\begin{align}\label{case3}
\mathcal{E}_m'\leq Ca_\gamma E_m^{\frac12}\+f\+_{\nu,m}^2,
\end{align}
and integrating it we obtain $\mathcal{E}_m(t)\leq \mathcal{E}_m(0)+C\sup_{[0,t]} E_m(s)^{\frac32}$. Since $\mathcal{E}_m\sim E_m$, we can apply the well-known argument to prove global-in-time existence for small $\mathcal{E}_m(0)$. We obtain the following result.
\begin{prop}\label{prop_global}
Let $m\geq 0$ be an integer, $f_0$ be an initial data of the equation \eqref{boltz} satisfying $F_0=\mu+\sqrt{\mu}f_0\geq 0$, and $a$ be a solution of the equation \eqref{scale_factor} satisfying $E_a\geq 0$. Then, there exist a constant $\varepsilon>0$ such that if $\mathcal{E}_m(0)<\varepsilon$, then the equation \eqref{boltz} admits a unique global-in-time solution $f$, corresponding to the initial data $f_0$, such that $F=\mu+\sqrt{\mu}f\geq 0$ and $\sup_{[0,\infty)}\mathcal{E}_m(t)\leq C\mathcal{E}_m(0)$.
\end{prop}

\subsection{Asymptotic behaviour}
Let $f$ be a solution of Proposition \ref{prop_global} satisfying $\mathcal{E}_m(t)\leq C\mathcal{E}_m(0)$, and we apply the interpolation introduced in \cite{StrainGuo06}. In the estimates \eqref{case21} and \eqref{case22}, we added all $0<k\leq m$, but if we only consider the summation over $0<k\leq r$ for an integer $0<r<m$, then the estimate \eqref{case3} can be rewritten as $y_r'/2+ a_\gamma\+f\+_{\nu,r}^2 \leq Ca_\gamma \mathcal{E}_m^{1/2}\+f\+_{\nu,r}^2$, and since $\mathcal{E}_m$ is small, we may write $y'_r+a_\gamma\+f\+^2_{\nu,r}\leq 0$. Note that $\+f\+_{\nu,r}$ is not equivalent to $\+f\+_{r}$, but only to $\+f\+_{(2r-1)/2}$. For an integer $k$ satisfying $r+k\leq m$, we use the interpolation as follows:
\begin{align*}
y_r(t)\leq C\+f(t)\+_r^2&\leq C\+f(t)\+_{r-1}^{2k/(k+1)}\+f(t)\+_{r+k}^{2/(k+1)}\\
&\leq C\+f(t)\+_{r-1}^{2k/(k+1)}\mathcal{E}_m^{1/(k+1)}(0)\leq C\+f(t)\+_{\nu,r}^{2k/(k+1)}\mathcal{E}_m^{1/(k+1)}(0).
\end{align*}
Applying this to the above differential inequality, we have $y'_r+C\mathcal{E}_m^{-1/k}(0)a_\gamma y_r^{(k+1)/k}\leq 0$, which can be written as $d(y^{-1/k}_r)/dt\geq Ca_\gamma\mathcal{E}_m^{-1/k}(0)$, and then integrate it to obtain
\begin{align}
y_r(t)\leq\bigg(y_r^{-\frac{1}{k}}(0)+C\mathcal{E}_m^{-\frac{1}{k}}(0)\int_0^t a_\gamma(s)ds\bigg)^{-k}\leq Cy_r(0)\bigg(1+\int_0^ta_\gamma(s)ds\bigg)^{-k}.\label{y_r}
\end{align}
Recall that $a_\gamma(t)=a^{-3-\gamma}(t)$, where $a(t)\leq C(1+t)^{2/3}$ for $E_a=0$ and $a(t)\leq C(1+t)$ for $E_a>0$. Note that $-3-\gamma$ is negative, and we have the following cases:
\begin{itemize}
\item[{\sf (i)}] If $E_a=0$ and $\gamma=-3/2$, then $a_\gamma(t)\geq C(1+t)^{-1}$ and
\begin{align*}
\int_0^t a_\gamma(s)ds\geq C\int_0^t(1+s)^{-1}ds\geq C\ln(1+t),
\end{align*}
hence we obtain from \eqref{y_r} that $y_r(t)\leq Cy_r(0)(1+\ln(1+t))^{-k}$.
\item[{\sf (ii)}] If $E_a=0$ and $-3<\gamma<-3/2$, then $a_\gamma(t)\geq C(1+t)^{-2-2\gamma/3}$ and
\begin{align*}
\int_0^ta_\gamma(s)ds&\geq C\int_0^t(1+s)^{-2-\frac23\gamma}ds\geq C\Big((1+t)^{-1-\frac23\gamma}-1\Big),
\end{align*}
therefore we have $y_r(t)\leq Cy_r(0)(1+t)^{k+2\gamma k/3}$. Note that $k+2\gamma k/3$ is negative for $k\geq 1$, hence $y_r$ decays.
\item[{\sf (iii)}] If $E_a>0$ and $\gamma=-2$, then we have as in {\sf (i)} that $y_r(t)\leq Cy_r(0)(1+\ln(1+t))^{-k}$.
\item[{\sf (iv)}] If $E_a>0$ and $-3<\gamma<-2$, then we have $y_r(t)\leq Cy_r(0)(1+t)^{2k+\gamma k}$ in a similar way. 
\end{itemize}
We combine these results and obtain the following theorem.

\begin{thm}
Let $m\geq 2$ be an integer, $f_0$ be an initial data of the equation \eqref{boltz} satisfying $F_0=\mu+\sqrt{\mu}f_0\geq 0$, and $a$ be a solution of the equation \eqref{scale_factor} satisfying $E_a\geq 0$. Suppose that $f$ is a solution constructed in Proposition \ref{prop_global} such that $\mathcal{E}_m(t)\leq C\mathcal{E}_m(0)$. If $r$ and $k$ are integers satisfying $0<r<r+k\leq m$, then $\+f(t)\+_r$ decays as follows:
\begin{align*}
&\+f(t)\+_r\leq C\mathcal{E}_m(0)(1+\ln(1+t))^{-k},\quad \mbox{if $E_a=0$ and $\gamma=-\frac32$,}\\
&\+f(t)\+_r\leq C\mathcal{E}_m(0)(1+t)^{k+\frac23\gamma k},\quad \mbox{if $E_a=0$ and $-3<\gamma<-\frac32$,}\allowdisplaybreaks\\
&\+f(t)\+_r\leq C\mathcal{E}_m(0)(1+\ln(1+t))^{-k},\quad \mbox{if $E_a>0$ and $\gamma=-2$,}\\
&\+f(t)\+_r\leq C\mathcal{E}_m(0)(1+t)^{2k+\gamma k},\quad \mbox{if $E_a>0$ and $-3<\gamma=-2$},
\end{align*}
where the constant $C$ depends on $E_a$, $\gamma$, and $k$.
\end{thm}
\subsection{Remarks}
In this paper we considered a universe model \eqref{universe}--\eqref{scale_factor}, which satisfies the VPB system in the soft potential case. To study its asymptotic stability, we considered a perturbation and assumed the spatial homogeneity to the transformed equations \eqref{vpb1_tilde}--\eqref{vpb2_tilde}. By the energy method of \cite{Guo031} and the interpolation of \cite{StrainGuo06} we obtained the global-in-time existence and the asymptotic behaviour of solutions of the equation \eqref{boltz}. It was observed in \cite{Lee131} that the expansion of the universe reduces the effect of collisions such that solutions do not seem to converge to the universe model \eqref{universe}--\eqref{scale_factor}, but in this paper we showed that it is recovered, in a certain range of soft potentials, such that the solutions asymptotically tend to the universe model. In the present paper we only considered the spatially homogeneous case, but it will be more challenging to study the inhomogeneous case. We hope that the result of this paper can provide an intuition on the research of the general relativistic Boltzmann equation.

\section*{Acknowledgements}
H. Lee has been supported by the TJ Park Science Fellowship of POSCO TJ Park Foundation. This research was supported by Basic Science Research Program through the National Research Foundation of Korea(NRF) funded by the Ministry of Science, ICT \& Future Planning(NRF-2015R1C1A1A01055216).


\begin{thebibliography}{99}

\bibitem{Glassey} Glassey, R. T.:
{\it The Cauchy Problem in Kinetic Theory.}
Society for Industrial and Applied Mathematics (SIAM), Philadelphia, PA, 1996.

\bibitem{Guo031} Guo, Y.:
Classical solutions to the Boltzmann equation for molecules with an angular cutoff.
{\it Arch. Ration. Mech. Anal.} 169 (2003), no. 4, 305--353.

\bibitem{Lee10} Lee, H.:
Classical solutions to the Vlasov-Poisson system in an accelerating cosmological setting.
{\it J. Differential Equations} 249 (2010), no. 5, 1111--1130.

\bibitem{Lee131} Lee, H.:
Global solutions of the Vlasov-Poisson-Boltzmann system in a cosmological setting.
{\it J. Math. Phys.} 54 (2013), no. 7, 073302, 15 pp.

\bibitem{Lee132} Lee, H.:
Asymptotic behaviour of the relativistic Boltzmann equation in the Robertson-Walker spacetime.
{\it J. Differential Equations} 255 (2013), no. 11, 4267--4288.

\bibitem{LeeNungesser} Lee, H., Nungesser, E.:
Future global existence and asymptotic behaviour of solutions to the Einstein-Boltzmann system with Bianchi I symmetry.
arXiv:1506.02440 [gr-qc]

\bibitem{LeeRendall131} Lee, H., Rendall, A. D.:
The Einstein-Boltzmann system and positivity.
{\it J. Hyperbolic Differ. Equ.} 10 (2013), no. 1, 77--104.

\bibitem{LeeRendall132} Lee, H., Rendall, A. D.:
The spatially homogeneous relativistic Boltzmann equation with a hard potential.
{\it Comm. Partial Differential Equations} 38 (2013), no. 12, 2238--2262.

\bibitem{Rein97} Rein, G.:
Nonlinear stability of homogeneous models in Newtonian cosmology.
{\it Arch. Rational Mech. Anal.} 140 (1997), no. 4, 335--351.

\bibitem{ReinRendall94} Rein, G., Rendall, A. D.:
Global existence of classical solutions to the Vlasov-Poisson system in a three dimensional,
cosmological setting.
{\it Arch. Ration. Mech. Anal.} 126 (1994), no. 2, 183--201.

\bibitem{StrainGuo06} Strain, R. M., Guo, Y.:
Almost exponential decay near Maxwellian.
{\it Comm. Partial Differential Equations} 31 (2006), no. 1-3, 417--429.







\end{thebibliography}
\end{document}